%% file: StoCacheExtConf-Clustered.tex
\documentclass[10pt, conference, a4paper]{IEEEtran}
\usepackage{amsmath,amsfonts,amssymb,amsbsy, amsthm,ae,aecompl}
\usepackage{algorithm,algpseudocode}
\usepackage[utf8]{inputenc}
\usepackage[english]{babel}
\usepackage{bm}
\usepackage{color}
\usepackage[noadjust]{cite}
\usepackage{epsfig}
\usepackage{enumerate}
\usepackage{float}
\usepackage{fancyhdr}
\usepackage[T1]{fontenc}
\usepackage[acronym,toc,shortcuts]{glossaries}
\usepackage{graphicx, caption, subcaption}
\usepackage{hyperref}
\usepackage{lipsum}
\usepackage{dsfont}
\usepackage{transparent}
\usepackage{tikz,pgfplots}
\usepackage{times}
\usepackage{verbatim}
\usepackage[all]{xy}
\hypersetup{
    colorlinks,%
    citecolor=black,%
    filecolor=black,%
    linkcolor=black,%
    urlcolor=black
}

\definecolor{light-gray}{gray}{0.95}

\usepackage{etoolbox}
\AtBeginEnvironment{align}{\setcounter{subeqn}{0}}
\newcounter{subeqn} %

\usetikzlibrary{calc}
\usetikzlibrary{arrows,decorations.markings}
\pgfplotsset{
  grid style = {
    dash pattern = on 0.025mm off 0.95mm on 0.025mm off 0mm, 
    line cap = round,
    black,
    line width = 0.5pt
  },
  tick label style={font=\small},
  label style={font=\small},
  legend style={font=\footnotesize},
}

\pgfplotscreateplotcyclelist{laneas-clustered1}{
	black, solid, thick, mark=none\\
	black, dashed, thick, mark=none\\
	black, dotted, thick, mark=none\\
	black!00!black, solid, thick, only marks, mark=+, mark size=2, every mark/.append style={solid,fill=gray!99!black}\\
	black, solid, thick, only marks, mark=star, mark size=2, every mark/.append style={solid,fill=black}\\
	black, solid, thick, only marks, mark=o, mark size=2, every mark/.append style={solid,fill=black}\\
}

\graphicspath{{figures/}}

\bgroup
\makeglossaries
\newacronym{ADMM}{ADMM}{Alternating Direction Method of Multipliers}
\newacronym{AWGN}{AWGN}{additive white Gaussian noise}
\newacronym{APC}{APC}{area power consumption}
\newacronym{ASE}{ASE}{area spectral efficiency}
\newacronym{BS}{BS}{base station}
\newacronym{CN}{CN}{core network}
\newacronym{CS}{CS}{central scheduler}
\newacronym{CR}{CR}{central router}
\newacronym{D2D}{D2D}{device-to-device}
\newacronym{iid}{i.i.d.}{independent and identically distributed}
\newacronym{EE}{EE}{energy efficiency}
\newacronym{LTE}{LTE}{long term evolution}
\newacronym{massive-MIMO}{massive-MIMO}{massive multiple-input multiple-output}
\newacronym{MBS}{MBS}{macro base station}
\newacronym{MIMO}{MIMO}{multiple-input multiple-output}
\newacronym{MU}{MU}{macro cell user}
\newacronym{PP}{PP}{point process}
\newacronym{PPP}{PPP}{{P}oisson point process}
\newacronym{RAN}{RAN}{radio access network}
\newacronym{SBS}{SBS}{small base station}
\newacronym{SNR}{SNR}{signal-to-noise ratio}
\newacronym{SINR}{SINR}{signal-to-interference-plus-noise ratio}
\newacronym{SIR}{SIR}{signal-to-interference ratio}
\newacronym{SCN}{SCN}{small cell network}
\newacronym{SU}{SU}{small cell user}
\newacronym{UT}{UT}{user terminal}
\newacronym{QoS}{QoS}{quality-of-service}
\newacronym{QoE}{QoE}{quality-of-experience}
\newacronym{PDF}{PDF}{probability distribution function}
\newacronym{PGFL}{PGFL}{probability generating functional}
\newacronym{HetNet}{HetNet}{heterogeneous network}

\newtheorem{definition}{Definition}
\newtheorem{theorem}{Theorem}

\newtheorem{remark}{Remark}

\begin{document}
\title{Edge Caching for Coverage and Capacity-aided Heterogeneous Networks}
\author{
		\IEEEauthorblockN{Ejder Baştuğ$^{\diamond}$, Mehdi Bennis$^{\dagger}$, Marios Kountouris$^{\circ}$ and Mérouane Debbah$^{\diamond, \circ}$\\}
		\IEEEauthorblockA{
				\small
				$^{\diamond}$Large Networks and Systems Group (LANEAS), CentraleSupélec, \\ Université Paris-Saclay, 3 rue Joliot-Curie,  91192 Gif-sur-Yvette, France\\	
				$^{\dagger}$Centre for Wireless Communications, University of Oulu, Finland \\
				$^{\circ}$Mathematical and Algorithmic Sciences Lab, Huawei France R\&D, Paris, France \\	
				\{ejder.bastug, merouane.debbah\}@centralesupelec.fr, marios.kountouris@huawei.com, bennis@ee.oulu.fi \\
				\vspace{-1.15cm}
		}
		\thanks{This research has been supported by the ERC Starting Grant 305123 MORE (Advanced Mathematical Tools for Complex Network Engineering), the project BESTCOM, the Academy of Finland CARMA project and TEKES grant (2364/31/2014).}
}
\IEEEoverridecommandlockouts
\maketitle	

\begin{abstract}
A two-tier heterogeneous cellular network (HCN) with intra-tier and inter-tier dependence is studied. The macro cell deployment follows a Poisson point process (PPP) and two different clustered point processes are used to model the cache-enabled small cells. Under this model, we derive approximate expressions in terms of finite integrals for the average delivery rate considering inter-tier and intra-tier dependence. On top of the fact that cache size drastically improves the performance of small cells in terms of average delivery rate, we show that rate splitting of limited-backhaul induces non-linear performance variations, and therefore has to be adjusted for rate fairness among users of different tiers. 
\end{abstract}
\section{Introduction}
Caching in heterogeneous cellular networks (HCN) significantly improves the system performance, and is of cardinal importance especially in limited-backhaul settings \cite{Bastug2014CacheEnabledExtendedArxiv, Shanmugam2013Femtocaching}. However, existing stochastic geometry-based analyses of caching in heterogeneous networks \cite{Yang2015Analysis,  Chen2016Cooperative}  ignore the impact of limited-backhaul and consider multiple tiers of mutually independent point processes. To remedy to this situation, this paper analyzes the benefits of edge caching in such wireless deployment conditions.

Consider a heterogeneous network consisting of mobile users, clustered cache-enabled \glspl{SBS}, \glspl{MBS} and \glspl{CR} for the backhaul. For \glspl{SBS}, we shall consider two different topologies: 1) coverage-aided and 2) capacity-aided deployments. Coverage-aided deployment follows a \emph{Poisson hole process} (PHP) while capacity-aided deployment is modeled using a \emph{Matérn cluster process} (MCP). This models inter-tier and intra-tier dependence, respectively. 
Due to the multi-tier dependence and the non-tractability of the considered point processes, an exact calculation of the interference and key performance metrics seems unfeasible. Hence, we restrict ourselves to approximations in the form of finite integrals and validate their behavior via numerical results. 

This hierarchical model exploiting random spanning trees hides the algorithmic details of such a complex network, yet allows for tractable mathematical expressions to characterize the overall performance of such an heterogeneous network. We show that while the average delivery rate of \glspl{SU} can be improved by adding more storage per \ac{SBS}, the backhaul rate splitting among tiers induces non-linear performance variations and requires adjustments for the rate fairness.
\section{System Model}
\label{sec:systemmodel}
The system model is composed of \glspl{MBS}, cache-enabled clustered \glspl{SBS}, \glspl{CR} and mobile users. We focus on two different topologies as  in \cite{Deng2014Heterogeneous}.
\subsection{Coverage-Aided Topology}
\label{sec:topologyI}
We model a heterogeneous cellular network which consists of \glspl{MBS} and cache-enabled \glspl{SBS}. The \glspl{MBS} are modeled by an independent homogeneous planar \ac{PPP} of intensity $\lambda_{\mathrm{mc}}$, denoted by $\Phi_{\mathrm{mc}} = \{y_i\}_{i \in \mathbb{N}}$ where $y_i$ denotes the location of the $i$-th \ac{MBS}. Additionally, the \emph{potential} \glspl{SBS} are located according to another independent homogeneous \ac{PPP} of intensity $\lambda_{\mathrm{sc'}}$, denoted by $\Phi_{\mathrm{sc'}} = \{x'_i\}_{i \in \mathbb{N}}$ where $x'_i$ represents the location of the $i$-th \ac{SBS}. We suppose that each \ac{MBS} has an exclusion region which is made of a disk with radius $R_{\mathrm{c}}$ centered at the position of \ac{MBS}. Assuming that the aim of \glspl{SBS} is to fill the coverage holes of \glspl{MBS} to provide a better service to users, these \glspl{SBS} are deployed outside of the exclusion regions of \glspl{MBS}. Therefore, the deployed \glspl{SBS} form clusters according to a \emph{Poisson hole process} (Cox process) as follows \cite{Haenggi2012Stochastic}.
	\begin{definition}[Clustering process of coverage-aided \glspl{SBS}]\label{def:smallCellCovPhole} 
	Let $\Phi_{\mathrm{mc}}$ be a homogeneous \ac{PPP} of intensity $\lambda_{\mathrm{mc}}$ for \glspl{MBS} and $\Phi_{\mathrm{sc'}}$ be an independent and homogeneous \ac{PPP} of intensity $\lambda_{\mathrm{sc'}}$ for potential \glspl{SBS}, with $\lambda_{\mathrm{sc'}} > \lambda_{\mathrm{mc}}$. For each $y \in \Phi_{\mathrm{mc}}$, remove all the points in
		\begin{equation}
			\Phi_{\mathrm{sc'}} \cap  \mathcal{B}(y, R_{\mathrm{c}})
		\end{equation}
		where $\mathcal{B}(y, R_{\mathrm{c}})$ is the ball of radius $R_{\mathrm{c}}$ centered at $y$. 	Then, the remaining points of $\Phi_{\mathrm{sc'}}$ form clusters, known as the Poisson hole process $\Phi_{\mathrm{sc}}$ and represent the deployed \glspl{SBS}. Moreover, this process has the intensity of 
		\begin{equation}
			\lambda_{\mathrm{sc}} = \lambda_{\mathrm{sc'}} \mathrm{exp}(-\lambda_{\mathrm{mc}} \pi R_{\mathrm{c}}^2).
		\end{equation}
	\end{definition}
\begin{figure}[ht!]
	\centering
	\includegraphics[width=0.70\linewidth]{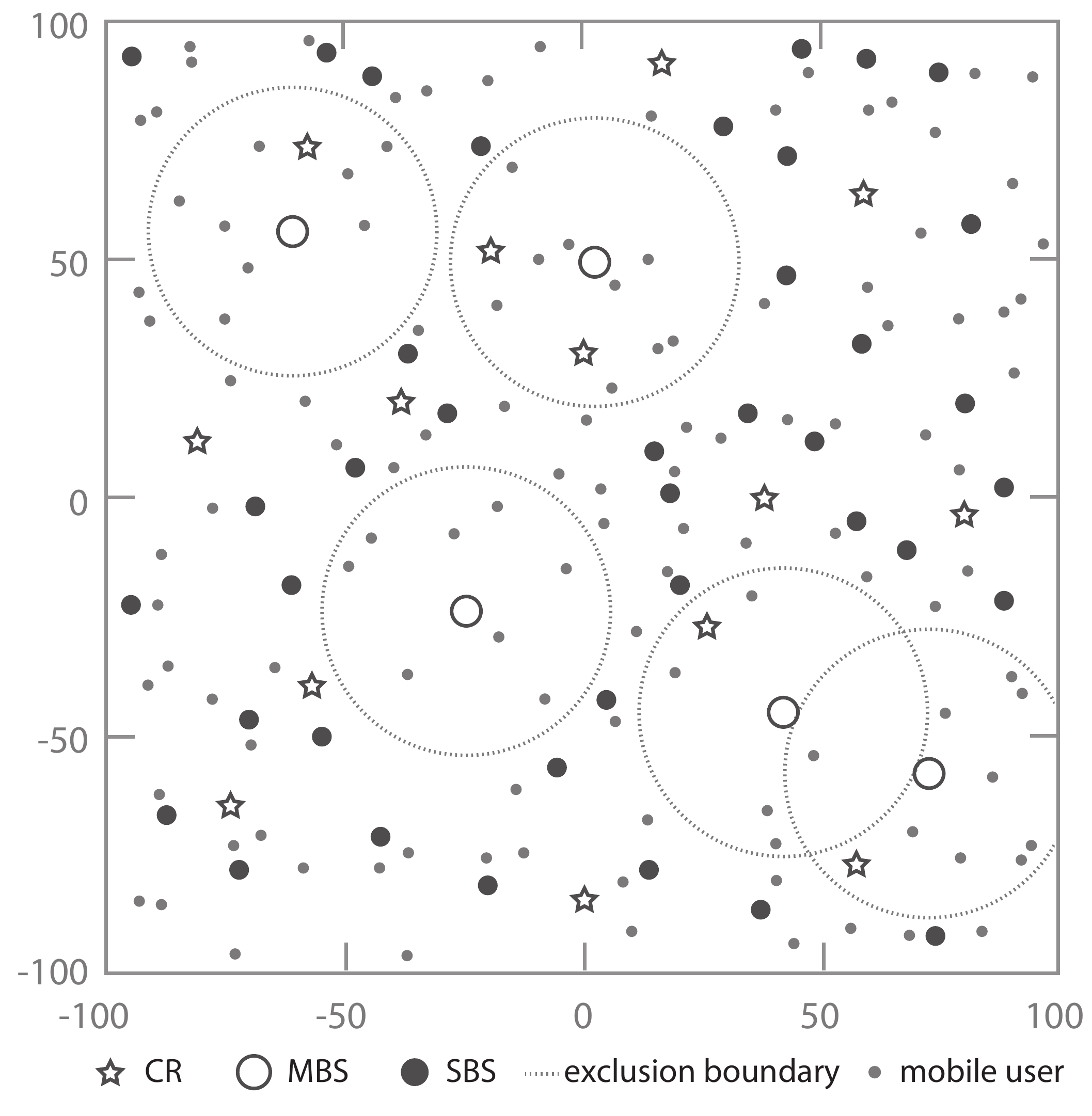}
	\caption{An illustration of the coverage-aided deployment.}
	\label{fig:scenario:cov}
	\vspace{-0.3cm}
\end{figure}
	On the other hand, \glspl{CR} are distributed in the plane according to an independent homogeneous \ac{PPP} of intensity $\lambda_{\mathrm{cr}}$, denoted by $\Phi_{\mathrm{cr}} = \{u_i\}_{i \in \mathbb{N}}$. These routers are in charge of providing broadband Internet connection to \glspl{MBS} and \glspl{SBS} via backhaul links. Mobile user terminals are also positioned in the whole plane according to an independent homogeneous \ac{PPP} of intensity $\lambda_{\mathrm{ut}}$, denoted by $\Phi_{\mathrm{ut}} = \{z_i\}_{i \in \mathbb{N}}$, with $\lambda_{\mathrm{ut}} > \lambda_{\mathrm{mc}}$ and $\lambda_{\mathrm{ut}} > \lambda_{\mathrm{sc}}$. The snapshots of these point processes are given in Fig. \ref{fig:scenario:cov}. 

\subsection{Capacity-Aided Topology}
\label{sec:topology2}
Suppose a two-tier heterogeneous cellular network consisting of \glspl{MBS} and cache-enabled \glspl{SBS}. The \glspl{MBS} are distributed according to an independent homogeneous \ac{PPP} of intensity $\lambda_{\mathrm{mc}}$, denoted by $\Phi_{\mathrm{mc}} = \{y_i\}_{i \in \mathbb{N}}$ where $y_i$ denotes the position of the $i$-th \ac{MBS}. The \glspl{SBS} are placed in hot-spots to sustain the demand of highly concentrated users, according to  an independent \emph{Matérn cluster process} $\Phi_{\mathrm{sc}} = \{x_i\}_{i \in \mathbb{N}}$ whose parent \ac{PPP} $\Phi_{\mathrm{sc'}}$ has intensity $\lambda_{\mathrm{sc'}}$. The process is given as follows \cite{Haenggi2012Stochastic}.
	\begin{definition}[Clustering process of capacity-aided \glspl{SBS}]\label{def:smallCellCapMatern}  
		Let $\Phi_{\mathrm{sc'}}$ be a parent process modelled by a homogeneous \ac{PPP} of intensity $\lambda_{\mathrm{sc'}}$. Then, the clustering process of \glspl{SBS} is given by
		\begin{equation}
			\Phi_{\mathrm{sc}} = \bigcup_{x' \in \Phi_{\mathrm{sc'}}} N^{x'}
		\end{equation}
		where $N^{x'}$ is a Poisson number of \ac{iid} points with mean $\bar{c}$, distributed uniformly in the ball $\mathcal{B}(x', R_{\mathrm{c}})$. Then, the process $\Phi_{\mathrm{sc}}$ is called  Matérn cluster process $\Phi_{\mathrm{sc}}$ and has intensity of 
		\begin{equation}
			\lambda_{\mathrm{sc}} = \lambda_{\mathrm{sc'}} \bar{c}.
		\end{equation}
	\end{definition}

\begin{figure}[ht!]
	\centering
	\includegraphics[width=0.70\linewidth]{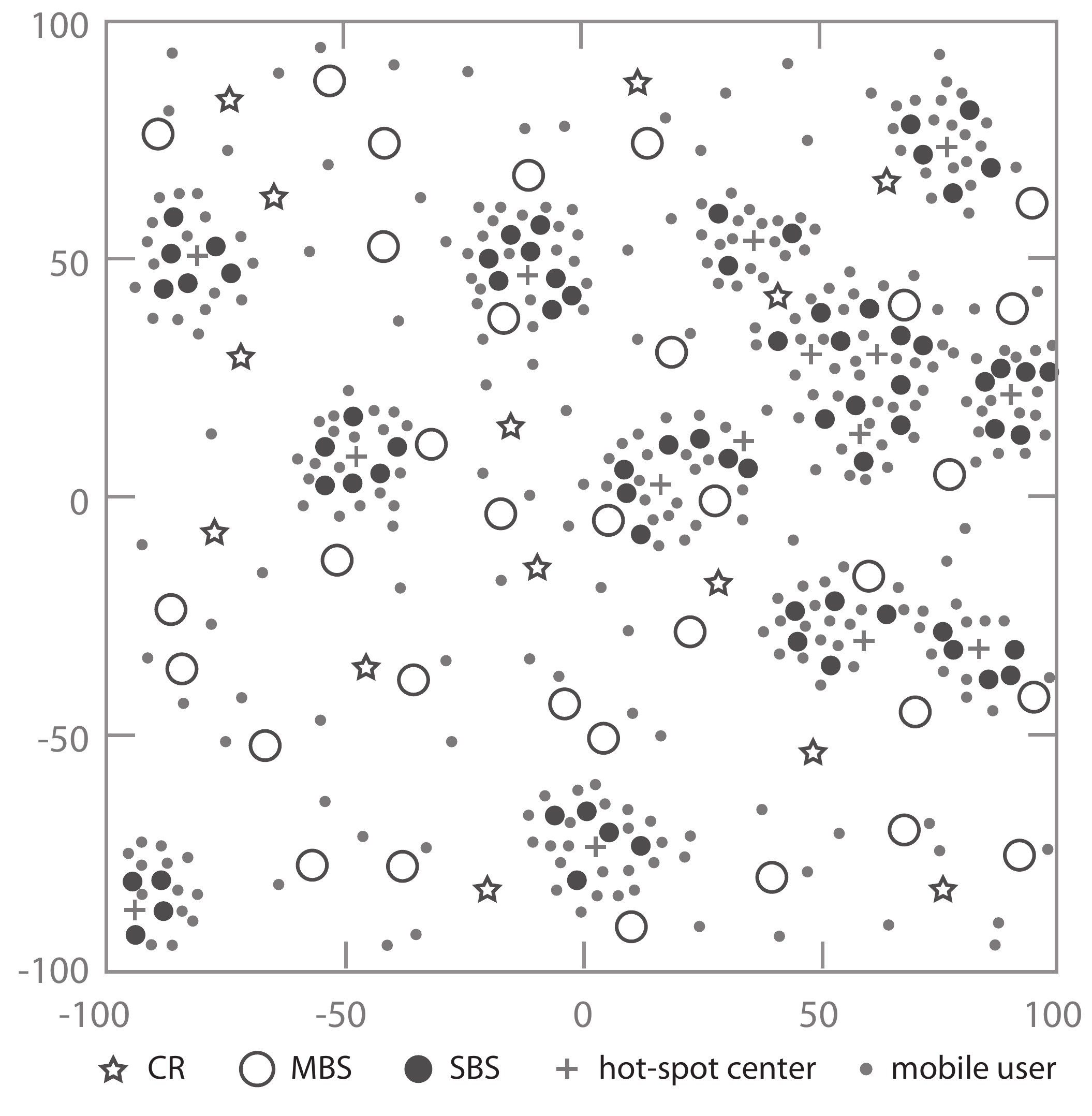}
	\caption{An illustration of the capacity-aided deployment.}
	\label{fig:scenario:cap}
	\vspace{-0.4cm}
\end{figure} 

For the users, we suppose that  mobile users (both \glspl{MU} and \glspl{SU}) are distributed on the two-dimensional Euclidean plane, however \glspl{SU} are highly concentrated in hot-spot regions served by \glspl{SBS}. From this motivation, we assume that all mobile users are distributed according to a Cox process $\Phi_{\mathrm{ut}} = \{z_i\}_{i \in \mathbb{N}}$ \cite{Haenggi2012Stochastic}. In particular, the population centers of radius $R_{\mathrm{c}}$ are drawn from the parent \ac{PPP} $\Phi_{\mathrm{sc'}}$, where \glspl{SU} in these clusters are distributed uniformly at random and are covered by \glspl{SBS} forming a Matern cluster process $\Phi_{\mathrm{sc}}$. By doing so, these mobile users are (on average) covered by their \glspl{SBS} deployed in these hot-spot areas.

Recalling $\bar{c}$ as the average number of \glspl{SBS} per cluster, the density of active \glspl{SU} per cluster is then $\lambda_{\mathrm{ut-s}} = \frac{\bar{c}}{\pi R_{\mathrm{c}}^2}$. The \glspl{MU} distributed in the rest of the network follow a \ac{PPP} with density $\lambda_{\mathrm{ut-m}}$ and are served by their own \glspl{MBS}. For convenience, we consider that each \ac{MBS} serves only one \ac{MU} on average and the same holds for each \ac{SBS} and its user. With this consideration in mind, the densities of \glspl{MU} and \glspl{SU} are equal to that of the \glspl{MBS} and \glspl{SBS}, respectively. Under this setting, the \glspl{MU} and \glspl{SU} form a Cox process with density $\lambda_{\mathrm{ut}} =  \lambda_{\mathrm{mc}} + \lambda_{\mathrm{sc}}$, clustered in hot-spots and uniformly distributed elsewhere. On the other hand, we consider that \glspl{CR} are modeled by an independent homogeneous \ac{PPP} of intensity $\lambda_{\mathrm{cr}}$, denoted by $\Phi_{\mathrm{cr}} = \{u_i\}_{i \in \mathbb{N}}$. A snapshot of the topology is depicted in Fig. \ref{fig:scenario:cap}. 
\subsection{Signal Model, Connectivity and Backhaul}
\label{sec:signalmodel}
The downlink transmissions of \glspl{MBS} and \glspl{SBS} occur at the same frequency with reuse factor $1$. The transmit power is $P_{\mathrm{mc}}$ for each \ac{MBS} and $P_{\mathrm{sc}}$ for each \ac{SBS}. All nodes (\glspl{MBS}, \glspl{SBS} and users) have single antenna. Having an \ac{MBS} positioned at $y$ and receiver at $z$ (or simply call as transmitter $y$ and user $z$), the channel coefficient is denoted by $h_{y,z} \in \mathbb{C}$. In case of  \ac{SBS} as a transmitter, the channel coefficient between transmitter $x$ and user $z$ is given by $g_{x,z} \in \mathbb{C}$. All the channel power coefficients are i.i.d. exponential random variables (Rayleigh fading) with $\mathbb{E}[|h_{y,z}|^2] = 1$ and $\mathbb{E}[|g_{x,z}|^2] = 1$. Supposing that the downlink rate of the typical user is a function of received \ac{SIR}, the target rate of signaling \ac{MBS} and \ac{SBS} are given by $\tau_{\mathrm{mc}}$ and $\tau_{\mathrm{sc}}$ respectively. 

Each user is either associated to the nearest \ac{MBS} or nearest \ac{SBS}. The backhaul connection of each base station is provided from its nearest \ac{CR}. Supposing that each \ac{CR} has a sufficiently high-capacity broadband Internet connection, \glspl{MBS} and \glspl{SBS} are connected to their nearest \glspl{CR} via error-free \emph{wired} backhaul links. In particular, each \ac{CR} has a total  capacity of $C_{\mathrm{cr}}$, which is an exponentially distributed random variable with mean $\mu$. 
\subsection{Caching}
\label{sec:caching}
We assume that the global content popularity distribution of users follow a power law defined as \cite{Newman2005Power}
\begin{equation}\label{eq:contentpdf}
f_{\mathrm{pop}}\left(f, \eta \right)
	=
	\begin{cases}
	\left(\eta - 1\right)f^{-\eta},
		& f \geq 1, \\
		0,			
		& f < 1,
	\end{cases}
\end{equation} 
where $f$ indicates a point in the support of the corresponding content and the parameter $\eta$ models the steepness of the distribution. 
In this work, we assume that  $f_{\mathrm{pop}}$ is perfectly known at the \glspl{SBS}. In $f_{\mathrm{pop}}\left(f,\eta\right)$, the contents in the interval $[1, F)$ are \emph{cacheable} and called as \emph{catalogue}, whereas the remaining part $[F, \infty]$ is called as non-cacheable contents (i.e., voice traffic, online gaming and sensor information). An interval $[f, f + \Delta f)$ in the support of $f_{\mathrm{pop}}\left(f,\eta\right)$ is considered as the probability of $f$-th content. Indeed, we assume that each content in the catalogue has a fixed length and called as \emph{chunk}. Each chunk can belong to a part of cacheable video file, audio or picture and so on. In fact, storing/distributing fixed-length chunks is one of the key principle in content centric networks \cite{Jacobson2009Networking}. Therefore, even though we use the term "content"  in the paper, the function $f_{\mathrm{pop}}\left(f,\eta\right)$  is actually a chunk popularity distribution. Each \ac{SBS} has a storage capacity of $F_{\mathrm{sc}}$ contents/chunks, with $1 \leq F_{\mathrm{sc}} \leq F$. 
\subsection{Hierarchical Model}
\label{sec:hierarchicalmodel}
The coverage and capacity-aided deployments of \glspl{SBS} together with \glspl{MBS}, users and \glspl{CR} can be modeled as random stationary graphs (hierarchical spanning trees) \cite{Suryaprakash2014Modeling}. In particular, a random hierarchical tree whose root node is a \ac{CR} located at $u$ is given by $\Psi = \{ (u, {\bf v}_u) \}$, where ${\bf v}_u$ is a mark vector containing all random variables associated with the \ac{CR}. In particular, the mark vector ${\bf v}_u$ contains information of \glspl{MBS} and \glspl{SBS} (with their users)  which are associated to the \ac{CR} at $u$, such as: 
\begin{itemize}
	\item $N_{\mathrm{mc}}$: The number of \glspl{MBS} connected to the \ac{CR} located at $u$. The vector ${\bf r}_{\mathrm{mc}} \in \mathbb{C}^{1 \times N_{\mathrm{mc}}}$ is the relative position vector of those \glspl{MBS} connected to the \ac{CR} at $u$. Therein, each element $r_{u,y}$ represents the distance from \ac{CR} at $u$ to \ac{MBS} at $y$. These positions are conditioned on $N_{\mathrm{mc}}$. 
	\item $N_{\mathrm{mu}}$: The number of users connected to an \ac{MBS} located at $y \in \{ {\bf r}_{\mathrm{mc}} \}$. The vector ${\bf r}_{\mathrm{mu}} \in \mathbb{C}^{1 \times N_{\mathrm{mu}}}$ is the relative positions of those \glspl{MU}  which are conditioned on  $N_{\mathrm{mu}}$. Each element $r_{y,z}$ represents distance from \ac{MBS} at $y$ to its user $z$.
	\item $N_{\mathrm{sc}}$: The number of \glspl{SBS} connected to the \ac{CR} located at $u$. The vector ${\bf r}_{\mathrm{sc}} \in \mathbb{C}^{1 \times N_{\mathrm{sc}}}$ is the relative positions of those \glspl{SBS} connected to the \ac{CR} at $u$. Here, each element $r_{u,x}$ represents the distance from \ac{CR} at $u$ to \ac{SBS} at $x$. These positions are conditioned on $N_{\mathrm{sc}}$.
	\item $N_{\mathrm{su}}$: The number of users connected to an \ac{SBS} located at $x \in \{ {\bf r}_{\mathrm{sc}} \}$. The vector ${\bf r}_{\mathrm{su}} \in \mathbb{C}^{1 \times N_{\mathrm{su}}}$ is the relative positions of users connected to the \ac{SBS} at $x \in \{ {\bf r}_{\mathrm{sc}} \}$. They are conditioned on  $N_{\mathrm{su}}$. Each element $r_{x,z}$ represents distance from \ac{SBS} at $x$ to its user at $z$.
\end{itemize}
\section{Performance Metrics}
\label{sec:performance}
For the performance metrics of coverage and capacity-aided deployments, we first start by defining \ac{SIR}.
\begin{definition}\label{def:sir}
The \ac{SIR} of a typical user connected to an \ac{MBS} (namely typical \ac{MU}) located at random position $y$ is defined as
\begin{align}
	\label{eq:sir:macro}
	\mathrm{SIR}_{\mathrm{mu}}
	&\triangleq
	\frac{ P_{\mathrm{mc}} h_{y} \ell(y) }{I_{\mathrm{mm}} + I_{\mathrm{sm}}}	
\end{align}
where $\ell(y) = \Vert y\Vert^{-\alpha}$ is the standard power-law path loss function (unless otherwise stated) with exponent $\alpha$, $I_{\mathrm{mm}} = \sum\limits_{y_{i} \in \Phi_{\mathrm{mc}} \setminus \{y\}} P_{\mathrm{mc}} h_{y_{i}} \ell(y_i)$ is the cumulative interference from other \glspl{MBS} except the serving cell at $y$, and $I_{\mathrm{sm}} = \sum\limits_{x_{i} \in \Phi_{\mathrm{sc}}} P_{\mathrm{sc}} g_{x_{i}} \ell(x_i)$ is the total interference from clustered \glspl{SBS}. Similarly, the \ac{SIR} of a typical user connected to an \ac{SBS} (namely typical \ac{SU}) located at random position $x$ is given by
\begin{align}
	\label{eq:sir:small}
	\mathrm{SIR}_{\mathrm{su}}
	&\triangleq
	\frac{ P_{\mathrm{mc}} g_{x} \ell(x) }{I_{\mathrm{ss}} + I_{\mathrm{ms}}}	
\end{align}
where $I_{\mathrm{ss}} = \sum\limits_{x_{i} \in \Phi_{\mathrm{sc}} \setminus \{ x \}} P_{\mathrm{sc}} g_{x_i} \ell(x_i)$ is the cumulative interference from other clustered \glspl{SBS} except the serving cell at $x$, and $I_{\mathrm{sm}} = \sum\limits_{y_{i} \in \Phi_{\mathrm{sc}}} P_{\mathrm{mc}} h_{y_i} \ell(y)$ is the total interference from \glspl{MBS}.
\end{definition}

The amount of backhaul rate allocated to typical users are defined by the following policy.
\begin{definition}[Backhaul Rate Splitting Policy]\label{def:backhaulpolicy}
Following the hierarchical model, suppose that a typical user located at $z \in \{ {\bf r}_{\mathrm{mu}} \}$  (typical \ac{MU}) is connected to the \ac{MBS} at $y \in \{ {\bf r}_{\mathrm{mc}} \}$, and this \ac{MBS} is connected to the nearest \ac{CR} at $u$. Then, the rate of backhaul link to the \ac{MBS} at $y$ is given as
\begin{equation}\label{eq:brate:macro}
	R'_{\mathrm{mu}} 	\triangleq \frac{\gamma C_{\mathrm{cr}} }
						  {\mathbb{E}\big[ N_{\mathrm{mc}} N_{\mathrm{mu}} \big]},
\end{equation} 
where $\gamma \in [0, 1]$ is a fraction of capacity allocated to the \glspl{MBS}. In case of \ac{SBS}, in a similar vein, a typical user located at $z \in \{ {\bf r}_{\mathrm{su}} \}$ is connected to the \ac{SBS} at $x \in \{ {\bf r}_{\mathrm{sc}} \}$ whose \ac{CR} is at $u$. Then, the rate of backhaul link to the \ac{SBS} is given as
\begin{equation}\label{eq:brate:small}
	R'_{\mathrm{su}} 	\triangleq \frac{(1 - \gamma) C_{\mathrm{cr}} }
						  {\mathbb{E}\big[ N_{\mathrm{sc}} N_{\mathrm{su}} \big]}.
\end{equation} 
\end{definition}

The expectation above is taken over the topology generated by the point processes of users and base stations. We now define our main performance metric as follows.
\begin{definition}[Delivery Rate]\label{def:deliveryrate}
The delivery rate of a typical user connected to an \ac{MBS} is defined as
\begin{align}
	\label{eq:deliveryrate:macro}
	R_{\mathrm{mu}}
	&\triangleq
	\begin{cases}
		\tau_{\mathrm{mc}}, 	& \text{if \;} 
											  	\mathrm{log}(1 + \mathrm{SIR}_{\mathrm{mu}}) > \tau_{\mathrm{mc}} 
											  	\text{\;and\;} 
											  	R'_{\mathrm{mu}} > \tau_{\mathrm{mc}}, \\											  	
		0,						& \text{otherwise}.
\end{cases}
\end{align}
Similarly, the delivery rate of a typical user connected to an \ac{SBS} is defined as
\begin{align}
	\label{eq:deliveryrate:small}
	R_{\mathrm{su}}
	&\triangleq
	\begin{cases}
		\tau_{\mathrm{sc}}, 	& \text{if \;} 
											  	\mathrm{log}(1 + \mathrm{SIR}_{\mathrm{su}}) > \tau_{\mathrm{sc}} 
											  	\text{\;and\;} 
											  	R'_{\mathrm{su}} > \tau_{\mathrm{sc}}, \\
		\tau_{\mathrm{sc}}, 	& \text{if \;} 
											  	\mathrm{log}(1 + \mathrm{SIR}_{\mathrm{su}}) > \tau_{\mathrm{sc}} 
											  	\text{\;and\;} 
											  	f_z \in \Delta_x, \\											
		0,						& \text{otherwise},
\end{cases}
\end{align}
where $f_z$ represents the content requested by the typical \ac{SU} and $\Delta_x $ is the cache of the \ac{SBS}.
\end{definition}

We are now ready to give the expressions for the average delivery rate of typical \glspl{MU} and \glspl{SU}.
\section{Main Results}
\begin{theorem}[Average Delivery Rate of Typical \ac{MU}]\label{theor:avgrate:macro}
The average delivery rate of a typical user connected to the nearest \ac{MBS} cell in coverage-aided deployment is approximated as
\begin{align}
	\bar{R}_{\mathrm{mu}}^{\mathrm{(cov)}} \approx 	
	\tau_{\mathrm{mc}}
	B^{\mathrm{(cov)}}_1
	B^{\mathrm{(cov)}}_2
\end{align}
where $B^{\mathrm{(cov)}}_1$ and $B^{\mathrm{(cov)}}_2$ are given in \eqref{eq:avgRateB1Cov} and \eqref{eq:avgRateB2Cov} respectively. Therein, Laplace transforms and other related function definitions are given below $B^{\mathrm{(cov)}}_1$, and $F(x,y;z;w)$ is the hypergeometric function. For capacity-aided deployment, we have
\begin{align}
	\bar{R}_{\mathrm{mu}}^{\mathrm{(cap)}} \approx
	\tau_{\mathrm{mc}}  
	B^{\mathrm{(cap)}}_1
	B^{\mathrm{(cap)}}_2
\end{align}
where $B^{\mathrm{(cap)}}_1$ and $B^{\mathrm{(cap)}}_2$ are given in \eqref{eq:avgRateB1Cap} and \eqref{eq:avgRateB2Cap} respectively. 
\end{theorem}
\begin{proof}
See Appendix C.1 in \cite{Bastug2015Distributed}.
\end{proof}
\begin{remark}
These expressions are cumbersome but numerically easy to compute. The terms $B_1$ and $B_2$ capture the downlink and backhaul behaviour respectively.
\end{remark}
\vspace{-0.3cm}
%
{\small
\begin{multline}
	B^{\mathrm{(cov)}}_1 
	= 
	\int_0^{R_{\mathrm{c}}}
		 e^{-\frac{( e^{\tau_{\mathrm{mc}}} - 1 )}{P_{\mathrm{mc}} r_{\mathrm{mc}}^{-\alpha}}}
		\mathcal{L}_{I_{\mathrm{mm}}}
		\Big(
			\frac{e^{\tau_{\mathrm{mc}}} - 1}{P_{\mathrm{mc}} r_{\mathrm{mc}}^{-\alpha}}
		\Big) \times \\
		\mathcal{L}_{I_{\mathrm{sm}}}
		\Big(
			\frac{e^{\tau_{\mathrm{mc}}} - 1}{P_{\mathrm{mc}} r_{\mathrm{mc}}^{-\alpha}}
		\Big)
		\frac{k}{\nu} \big( \frac{r_{\mathrm{mc}}}{\nu}\big)^{k - 1} e^{-(r_{\mathrm{mc}}/\nu)^k}
		\mathrm{d}r_{\mathrm{mc}} \label{eq:avgRateB1Cov} 
\end{multline}
\vspace{-0.5cm}
\begin{multline}
	\mathcal{L}_{I_{\mathrm{mm}}}(s) 
	=
	\mathrm{exp}
	\Big(
		\frac{-s \pi \lambda_{\mathrm{mc}} P_{\mathrm{mc}} (2/\alpha)}{1 - 2/\alpha}
		r_{\mathrm{mc}}^{2 - \alpha} \times \\
		F
		\big(
		1,
		1 - 2/\alpha;
		2 - 2/\alpha;
		-s P_{\mathrm{mc}} r_{\mathrm{mc}}^{-\alpha}
		\big)
	\Big) \label{eq:avgRateLmmCov}
\end{multline}
\vspace{-0.4cm}
\begin{multline}
	\mathcal{L}_{I_{\mathrm{sm}}}(s)
	= 
	\mathrm{exp}
	\Big\{
		-\lambda_{\mathrm{sc}'}
		\Big(
			\frac{
				(s P_{\mathrm{sc}})^{2/\alpha}
				\pi^2
				(2/\alpha)
			}{\mathrm{sin}(\pi \frac{2}{\alpha})}
			- \\
			\pi
			R_{\mathrm{c}}^2
			A_{\mathrm{mc}}(s,R_{\mathrm{c}})
		\Big)
	\Big\}	\label{eq:avgRateLsmCov} 
\end{multline}
\vspace{-0.4cm}
\begin{multline}
	A_{\mathrm{mc}}(s,R_{\mathrm{c}})
	=
	\frac{1}{\pi R_{\mathrm{c}}^2}
	\int_0^{2\pi}
	\int_0^{r_{\mathrm{mc}} \mathrm{cos}	\varphi + \sqrt{ R_{\mathrm{c}}^2 - r_{\mathrm{mc}}^2 \mathrm{sin}^{2}\varphi}} \\
	{
		\frac{r \mathrm{d}r \mathrm{d} \varphi}
		{
			1 + s^{-1} P_{\mathrm{sc}}^{-1} r^{\alpha}		
		}
	}	\label{eq:avgRateAmcCov}
\end{multline}
\vspace{-0.4cm}
\begin{multline}
	B^{\mathrm{(cov)}}_2 
	= 
	1 -  \mathrm{exp}
		\Big(
			-\frac{
				\tau_{\mathrm{mc}}
				\lambda_{\mathrm{cr}} 
				\big(\lambda_{\mathrm{mc}} + \lambda_{\mathrm{sc}'} \mathrm{exp}(- \lambda_{\mathrm{mc}} \pi R_{\mathrm{c}}^2) \big)
			}{
				\mu 
				\gamma 
				\lambda_{\mathrm{mc}}^2
				\lambda_{\mathrm{ut}}
			} 
		\Big) \label{eq:avgRateB2Cov} 
\end{multline}
\vspace{-0.4cm}
}

{\small
\begin{multline}
	B^{\mathrm{(cap)}}_1 = 
	\int_0^{R_{\mathrm{c}}}
		 e^{-\frac{( e^{\tau_{\mathrm{mc}}} - 1 )}{P_{\mathrm{mc}} r_{\mathrm{mc}}^{-\alpha}}}
		\mathcal{L}_{I_{\mathrm{mm}}}
		\Big(
			\frac{e^{\tau_{\mathrm{mc}}} - 1}{P_{\mathrm{mc}} r_{\mathrm{mc}}^{-\alpha}}
		\Big) \times \\
		\mathcal{L}_{I_{\mathrm{sm}}}
		\Big(
			\frac{e^{\tau_{\mathrm{mc}}} - 1}{P_{\mathrm{mc}} r_{\mathrm{mc}}^{-\alpha}}
		\Big)
		\frac{k}{\nu} \big( \frac{r_{\mathrm{mc}}}{\nu}\big)^{k - 1} e^{-(r_{\mathrm{mc}}/\nu)^k}
		\mathrm{d}r_{\mathrm{mc}} \label{eq:avgRateB1Cap}
\end{multline}
\vspace{-0.4cm}
\begin{multline}
	\mathcal{L}_{I_{\mathrm{mm}}}(s)
	=
	\mathrm{exp}
	\Big(
		\frac{-s \pi \lambda_{\mathrm{mc}} P_{\mathrm{mc}} (2/\alpha)}{1 - 2/\alpha}
		r_{\mathrm{mc}}^{2 - \alpha} \times \\
		F
		\big(
		1,
		1 - 2/\alpha;
		2 - 2/\alpha;
		-s P_{\mathrm{mc}} r_{\mathrm{mc}}^{-\alpha}
		\big)
	\Big)	\label{eq:avgRateLmmCap}
\end{multline}
\vspace{-0.4cm}
\begin{multline}
	\mathcal{L}_{I_{\mathrm{sm}}}(s)
	=
	\mathrm{exp}
	\Big(
		- \lambda_{\mathrm{sc}'}
		\int_{\mathbb{R}^2}
		\Big(
			1 - 
			\mathrm{exp}(
				-\bar{c} \nu(s, y)
			)
		\Big)
		\mathrm{d}
		y
	\Big) \label{eq:avgRateLsmCap}
\end{multline}
\vspace{-0.4cm}
\begin{equation}
	\nu(s, y) 
	= 
	\int_{\mathbb{R}^2}{\frac{f(x)}{1 + (s P_{\mathrm{sc}} \ell(x - y) )^{-1}} \mathrm{d}x}
\end{equation}
\vspace{-0.1cm}
\begin{equation}
	f(x) 
	=
	\begin{cases}
		\frac{1}{\pi R_{\mathrm{c}}^2}, 	& \text{if \;} 
											  	\Vert x \Vert < R_{\mathrm{c}}, \\				  	
		0,						& \text{otherwise}.
	\end{cases}
\end{equation}
\vspace{-0.1cm}
\begin{equation}	
	B^{\mathrm{(cap)}}_2 
	= 
	1 - \mathrm{exp}
		\Big(
			-\frac{
				\tau_{\mathrm{mc}}
				\lambda_{\mathrm{cr}}
			}{
				\mu 
				\gamma 
				\lambda_{\mathrm{ut-m}}
			} 
		\Big) \label{eq:avgRateB2Cap} 
\end{equation}
\vspace{-0.4cm}	
}

\begin{theorem}[Average Delivery Rate of Typical \ac{SU}]\label{theor:avgrate:small}
The average delivery rate of the typical user connected to the nearest \ac{SBS} in coverage-aided deployment is approximated as
\begin{align}
	\bar{R}_{\mathrm{su}}^{\mathrm{(cov)}} \approx  
	\tau_{\mathrm{sc}}
	C^{\mathrm{(cov)}}_1
	C^{\mathrm{(cov)}}_2
	+
	\tau_{\mathrm{sc}}
	C^{\mathrm{(cov)}}_1
	C^{\mathrm{(cov)}}_3 \\
	- 
	\tau_{\mathrm{sc}}
	C^{\mathrm{(cov)}}_1
	C^{\mathrm{(cov)}}_2
	C^{\mathrm{(cov)}}_3
\end{align}
where $C^{\mathrm{(cov)}}_1$, $C^{\mathrm{(cov)}}_2$ and $C^{\mathrm{(cov)}}_3$ are given in \eqref{eq:avgRateC1Cov}, \eqref{eq:avgRateC2Cov} and \eqref{eq:avgRateC3Cov} respectively. Therein, Laplace transforms and other related function definitions are given below $C^{\mathrm{(cov)}}_1$. For capacity-aided deployment, we have
\begin{align}
	\bar{R}_{\mathrm{su}}^{\mathrm{(cap)}} \approx  
	\tau_{\mathrm{sc}}
	C^{\mathrm{(cap)}}_1
	C^{\mathrm{(cap)}}_2
	+
	\tau_{\mathrm{sc}}
	C^{\mathrm{(cap)}}_1
	C^{\mathrm{(cap)}}_3 \\ 
	-
	\tau_{\mathrm{sc}}
	C^{\mathrm{(cap)}}_1
	C^{\mathrm{(cap)}}_2
	C^{\mathrm{(cap)}}_3
\end{align}
where $C^{\mathrm{(cap)}}_1$, $C^{\mathrm{(cap)}}_2$ and $C^{\mathrm{(cap)}}_3$ are given in \eqref{eq:avgRateC1Cap}, \eqref{eq:avgRateC2Cap} and \eqref{eq:avgRateC3Cap} respectively. 
\end{theorem}
\begin{proof}
See Appendix C.2 in \cite{Bastug2015Distributed}.
\end{proof}
\begin{remark}
The terms $C_1$, $C_2$ and $C_3$ incorporate the downlink,  backhaul and caching aspects respectively.
\end{remark}
%
{\small
\begin{multline}
	C^{\mathrm{(cov)}}_1 =
	\int_0^{R_{\mathrm{c}}}
		 e^{-\frac{( e^{\tau_{\mathrm{sc}}} - 1 )}{P_{\mathrm{sc}} r_{\mathrm{sc}}^{-\alpha}}}
		\mathcal{L}_{I_{\mathrm{ss}}}
		\Big(
			\frac{e^{\tau_{\mathrm{sc}}} - 1}{P_{\mathrm{sc}} r_{\mathrm{sc}}^{-\alpha}}
		\Big) \times \\
		\mathcal{L}_{I_{\mathrm{ms}}}
		\Big(
			\frac{e^{\tau_{\mathrm{sc}}} - 1}{P_{\mathrm{sc}} r_{\mathrm{sc}}^{-\alpha}}
		\Big)
		\frac{k}{\nu} \big( \frac{r_{\mathrm{sc}}}{\nu}\big)^{k - 1} e^{-(r_{\mathrm{sc}}/\nu)^k}
		\mathrm{d}r_{\mathrm{sc}} \label{eq:avgRateC1Cov}
\end{multline}
\vspace{-0.4cm}	
\begin{multline}
	\mathcal{L}_{I_{\mathrm{ss}}}(s)
	=
	\mathrm{exp}
	\Big(
		\frac{-s \pi \lambda_{\mathrm{sc}'} P_{\mathrm{sc}} (2/\alpha)}{1 - 2/\alpha}
		r_{\mathrm{sc}}^{2 - \alpha} \times \\
		F
		\big(
		1,
		1 - 2/\alpha;
		2 - 2/\alpha;
		-s P_{\mathrm{sc}} r_{\mathrm{sc}}^{-\alpha}
		\big)
	\Big)  \label{eq:avgRateC1LssCov}
\end{multline}
\vspace{-0.4cm}	
\begin{multline}
	\mathcal{L}_{I_{\mathrm{ms}}} (s)
	= 
	\mathrm{exp}
	\Big\{
		-\lambda_{\mathrm{mc}}
		\Big(
			\frac{
				(s P_{\mathrm{mc}})^{2/\alpha}
				\pi^2
				(2/\alpha)
			}{\mathrm{sin}(\pi \frac{2}{\alpha})}
			-  \\
			\pi
			R_{\mathrm{c}}^2
			A_{\mathrm{sc}}(s,R_{\mathrm{c}})
		\Big)
	\Big\}	\label{eq:avgRateC1LmsCov}
\end{multline}
\vspace{-0.4cm}	
\begin{multline}
	A_{\mathrm{sc}}(s,R_{\mathrm{c}})
	=
	\frac{1}{\pi R_{\mathrm{c}}^2}
	\int_0^{2\pi}
	\int_0^{r_{\mathrm{sc}} \mathrm{cos}	\varphi + \sqrt{ R_{\mathrm{c}}^2 - r_{\mathrm{sc}}^2 \mathrm{sin}^{2}\varphi}} \\
	{
		\frac{r \mathrm{d}r \mathrm{d} \varphi}
		{
			1 + s^{-1} P_{\mathrm{mc}}^{-1} r^{\alpha}		
		}
	} \label{eq:avgRateC1AscCov}
\end{multline}
\vspace{-0.1cm}	
\begin{equation}
	C^{\mathrm{(cov)}}_2 
	=
	1 - \mathrm{exp}
		\Big(
			-\frac{
				\tau_{\mathrm{sc}}
				\lambda_{\mathrm{cr}}
				(\lambda_{\mathrm{mr}} +  \lambda_{\mathrm{sc}})
			}{
				\mu 
				\gamma 
				\lambda_{\mathrm{sc}}^2 \lambda_{\mathrm{ut}}
			} 
		\Big) \label{eq:avgRateC2Cov}
\end{equation}
\vspace{-0.1cm}	
\begin{equation}
	C^{\mathrm{(cov)}}_3  =
	1 - \big(1 + F_{\mathrm{sc}} \big)^{1 - \eta} \label{eq:avgRateC3Cov}
\end{equation}	
}
%
{\small
\begin{multline}
	C^{\mathrm{(cap)}}_1 
	=
	\int_0^{R_{\mathrm{c}}}
		 e^{-\frac{( e^{\tau_{\mathrm{sc}}} - 1 )}{P_{\mathrm{sc}} r_{\mathrm{sc}}^{-\alpha}}}
		\mathcal{L}_{I_{\mathrm{ss}}}
		\Big(
			\frac{e^{\tau_{\mathrm{sc}}} - 1}{P_{\mathrm{sc}} r_{\mathrm{sc}}^{-\alpha}}
		\Big) \times \\
		\mathcal{L}_{I_{\mathrm{ms}}}
		\Big(
			\frac{e^{\tau_{\mathrm{sc}}} - 1}{P_{\mathrm{sc}} r_{\mathrm{sc}}^{-\alpha}}
		\Big)
		\frac{k}{\nu} \big( \frac{r_{\mathrm{sc}}}{\nu}\big)^{k - 1} e^{-(r_{\mathrm{sc}}/\nu)^k}
		\mathrm{d}r_{\mathrm{sc}} \label{eq:avgRateC1Cap}
\end{multline}
\vspace{-0.4cm}	
\begin{multline}
	\mathcal{L}_{I_{\mathrm{ss}}}(s)
	=
	\mathrm{exp}
	\Big(
		- \lambda_{\mathrm{sc}'}
		\int_{\mathbb{R}^2}
		\Big(
			1 - 
			\mathrm{exp}(
				-\bar{c} \nu(s, x)
			)
		\Big)
		\mathrm{d}
		x
	\Big) \\
	\int_{\mathbb{R}^2}
	\Big(
		\mathrm{exp}(
			-\bar{c} \nu(s, x)
		)
	\Big)
		f(x)
		\mathrm{d}
		x		\label{eq:avgRateLssCap}
\end{multline}
\vspace{-0.1cm}	
\begin{equation}	
	\nu(s, x) 
	= 
	\int_{\mathbb{R}^2}{\frac{f(y)}{1 + (s P_{\mathrm{sc}} \tilde{\ell}(y - x) )^{-1}} \mathrm{d}y}
\end{equation}
\vspace{-0.1cm}	
\begin{equation}
	\mathcal{L}_{I_{\mathrm{ms}}}(s)
	=
	\mathrm{exp}
	\Big(
		-\lambda_{\mathrm{mc}}
		\frac{
				(s P_{\mathrm{mc}})^{2/\alpha}
				\pi^2
				(2/\alpha)
		}{
			\mathrm{sin}(\pi \frac{2}{\alpha})
		}
	\Big) \label{eq:avgRateLmsCap}
\end{equation}
\vspace{-0.1cm}	
\begin{equation}	
	C^{\mathrm{(cap)}}_2 
	=
	1 - \mathrm{exp}
		\Big(
			-\frac{
				\tau_{\mathrm{mc}}
				\lambda_{\mathrm{cr}}
			}{
				\mu 
				\gamma 
				\lambda_{\mathrm{ut-s}}
			} 
		\Big) \label{eq:avgRateC2Cap}
\end{equation}
\vspace{-0.1cm}	
\begin{equation}
	C^{\mathrm{(cap)}}_3
	=
	1 - \big(1 + F_{\mathrm{sc}} \big)^{1 - \eta} \label{eq:avgRateC3Cap}
\end{equation}	
\vspace{-0.4cm}
}
%
\section{Numerical Results and Discussion}\label{sec:validation}
In this section, we validate our expressions via Monte-Carlo simulations. First, the impact of backhaul rate splitting ratio $\gamma$ on the average delivery rate is given in Figs. (\ref{fig:avgDeliveryRate:cov}a) and (\ref{fig:avgDeliveryRate:cap}a) for coverage and capacity-aided deployments respectively. Indeed, as seen from the figures, a dramatical increase in average delivery rate occurs which confirms our intuitions. The rate of decrement for the typical \ac{SU} is relatively slow compared to the typical \ac{MU}. In order to find a balance between average delivery rate of typical \glspl{MU} and \glspl{SU}, ensuring rate fairness, the plots show that one has to set $\gamma$ carefully. In all of these cases, we observe that having caching capabilities at \glspl{SBS} improves system performance in terms of average delivery rate. In other words, a heterogeneous network consists of \glspl{MBS} and cache-enabled \glspl{SBS} allows higher average delivery rates while ensuring fairness between users at different tiers.

Second, the impact of storage size $F_{\mathrm{sc}}$ on the average delivery rate is given in Figs. (\ref{fig:avgDeliveryRate:cov}b) and (\ref{fig:avgDeliveryRate:cap}b) for coverage and capacity-aided deployments respectively. It is shown that increasing storage size of \glspl{SBS} both in coverage and capacity-aided deployments yields higher average delivery rates. The increment of storage size in coverage-aided deployment is more visible compared to capacity-aided deployment and allows typical \glspl{SU} to achieve higher rates than typical \glspl{MU}. Given the fact that caching  monotonically improves the overall system performance of \glspl{SU}, more storage does not seem necessary if one considers a linear cost for storage size, even though more storage is desirable for improving the average delivery rate.
\input{results-mainCov}
\input{results-mainCap}
\bibliographystyle{IEEEtran}
\bibliography{references}
\end{document}

%% file: results-mainCov.tex
\begin{figure}[ht!]
\centering
\begin{tikzpicture}[scale=0.51, baseline]
	\begin{axis}[
    every tick label/.append style  =
    { 
        font=\large
    },
	  yticklabel style={
				/pgf/number format/fixed zerofill,
				/pgf/number format/precision=2
	  },
 		grid = major,
 		cycle list name={laneas-clustered1},
 		mark repeat={1},		
	  legend columns=3,
		legend entries={
			MU The.: $\bar{R}_{\mathrm{mu}}^{\mathrm{(cov)}}$,
			SU The. (No Cache): $\bar{R}_{\mathrm{su}}^{\mathrm{(cov)}}$,
			SU The.: $\bar{R}_{\mathrm{su}}^{\mathrm{(cov)}}$,
		  MU Sim.: $\bar{R}_{\mathrm{mu}}^{\mathrm{(cov)}}$,
		  SU Sim.  (No Cache): $\bar{R}_{\mathrm{su}}^{\mathrm{(cov)}}$,
		  SU Sim.: $\bar{R}_{\mathrm{su}}^{\mathrm{(cov)}}$
		},
		legend cell align=left,
	  legend to name=namedmethods4stogeoccov,
	  	legend style={font=\scriptsize},
		label style={font=\large},
 		xlabel={{\Large (a)} Rate splitting ratio $\gamma$},
		]

 		\addplot+table [col sep=comma] {\string"rateSplit-theoryCovAvgRateDeliveryMu.csv"};
 		\addplot+table [col sep=comma] {\string"rateSplit-theoryCovAvgRateDeliverySuNoCache.csv"};
 		\addplot+table [col sep=comma] {\string"rateSplit-theoryCovAvgRateDeliverySu.csv"};

 		\addplot+table [col sep=comma]{\string"rateSplit-simCovAvgRateDeliveryMu.csv"};		
 		\addplot+table [col sep=comma]{\string"rateSplit-simCovAvgRateDeliverySuNoCache.csv"};	 
 		\addplot+table [col sep=comma]{\string"rateSplit-simCovAvgRateDeliverySu.csv"}; 							  			  		
	\end{axis}
\end{tikzpicture}
\begin{tikzpicture}[scale=0.51, baseline]
	\begin{axis}[
    every tick label/.append style  =
    { 
        font=\large
    },
	  yticklabel style={
				/pgf/number format/fixed zerofill,
				/pgf/number format/precision=2
	  },
 		grid = major,
 		cycle list name={laneas-clustered1},
 		legend cell align=left,
 		mark repeat={1},		
 		legend style ={legend pos=north east},
		label style={font=\large},
 		xlabel={{\Large (b)} Storage size $F_{\mathrm{sc}}$},
		]

 		\addplot+table [col sep=comma] {\string"storageSz-theoryCovAvgRateDeliveryMu.csv"};
 		\addplot+table [col sep=comma] {\string"storageSz-theoryCovAvgRateDeliverySuNoCache.csv"};
 		\addplot+table [col sep=comma] {\string"storageSz-theoryCovAvgRateDeliverySu.csv"};

 		\addplot+table [col sep=comma]{\string"storageSz-simCovAvgRateDeliveryMu.csv"};		
 		\addplot+table [col sep=comma]{\string"storageSz-simCovAvgRateDeliverySuNoCache.csv"};	 
 		\addplot+table [col sep=comma]{\string"storageSz-simCovAvgRateDeliverySu.csv"}; 							  			  		
	\end{axis}
\end{tikzpicture}
\vspace*{0.0cm}
\\
\scriptsize
\hspace{0.65cm}\ref{namedmethods4stogeoccov}
\caption{\small Evolution of average delivery rate in coverage-aided deployment.
$\lambda_{\mathrm{cr}} = 1.0\times 10^{-5}$, $\lambda_{\mathrm{mc}} = 1.5\times 10^{-5}$, $\lambda_{\mathrm{sc}'} = 5.5\times 10^{-5}$, $\lambda_{\mathrm{ut}} = 12.8\times 10^{-5}$ unit/m$^2$;
$P_{\mathrm{mc}} = 16$, $P_{\mathrm{sc}} = 3$ Watt;
$\tau_{\mathrm{mc}} = \tau_{\mathrm{sc}} = 4$ bits/s/Hz;
$\alpha = 4$	;
$R_{\mathrm{c}} = 80$ meters;
$\mu = 30$ bits/s/Hz; 
$\gamma = 0.6$;
$f_0 = 500$ GByte;
$F_{\mathrm{sc}} = 4$ GByte;
$\eta = 1.45$.}
\label{fig:avgDeliveryRate:cov}
\end{figure}
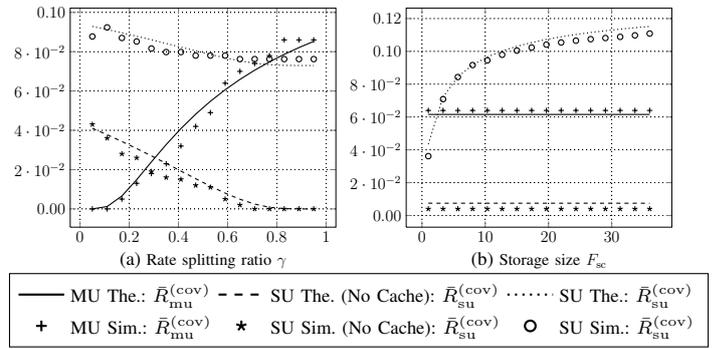

%% file: results-mainCap.tex
\begin{figure}[ht!]
\centering
\begin{tikzpicture}[scale=0.51, baseline]
	\begin{axis}[
    every tick label/.append style  =
    { 
        font=\large
    },
	  yticklabel style={
				/pgf/number format/fixed zerofill,
				/pgf/number format/precision=2
	  },
 		grid = major,
 		cycle list name={laneas-clustered1},
 		mark repeat={1},		
	  legend columns=3,
		legend entries={
			MU The.: $\bar{R}_{\mathrm{mu}}^{\mathrm{(cap)}}$,
			SU The. (No Cache): $\bar{R}_{\mathrm{su}}^{\mathrm{(cap)}}$,
			SU The.: $\bar{R}_{\mathrm{su}}^{\mathrm{(cap)}}$,
		  MU Sim.: $\bar{R}_{\mathrm{mu}}^{\mathrm{(cap)}}$,
		  SU Sim.  (No Cache): $\bar{R}_{\mathrm{su}}^{\mathrm{(cap)}}$,
		  SU Sim.: $\bar{R}_{\mathrm{su}}^{\mathrm{(cap)}}$
		},
		legend cell align=left,
	    legend to name=namedmethods4stogeoccov,
	    legend style={font=\scriptsize},
		label style={font=\large},
 		xlabel={{\Large (a)} Rate splitting ratio  $\gamma$},
		]

 		\addplot+table [col sep=comma] {\string"rateSplit-theoryCapAvgRateDeliveryMu.csv"};
 		\addplot+table [col sep=comma] {\string"rateSplit-theoryCapAvgRateDeliverySuNoCache.csv"};
 		\addplot+table [col sep=comma] {\string"rateSplit-theoryCapAvgRateDeliverySu.csv"};

 		\addplot+table [col sep=comma]{\string"rateSplit-simCapAvgRateDeliveryMu.csv"};		
 		\addplot+table [col sep=comma]{\string"rateSplit-simCapAvgRateDeliverySuNoCache.csv"};	 
 		\addplot+table [col sep=comma]{\string"rateSplit-simCapAvgRateDeliverySu.csv"}; 
						  			  		
	\end{axis}
\end{tikzpicture}
\begin{tikzpicture}[scale=0.51, baseline]
	\begin{axis}[
    every tick label/.append style  =
    { 
        font=\large
    },
	  yticklabel style={
				/pgf/number format/fixed zerofill,
				/pgf/number format/precision=2
	  },
 		grid = major,
 		cycle list name={laneas-clustered1},
 		legend cell align=left,
 		mark repeat={1},		
 		legend style ={legend pos=north east},
		label style={font=\large},
 		xlabel={{\Large (b)} Storage size $F_{\mathrm{sc}}$},
		]

 		\addplot+table [col sep=comma] {\string"storageSz-theoryCapAvgRateDeliveryMu.csv"};
 		\addplot+table [col sep=comma] {\string"storageSz-theoryCapAvgRateDeliverySuNoCache.csv"};
 		\addplot+table [col sep=comma] {\string"storageSz-theoryCapAvgRateDeliverySu.csv"};

 		\addplot+table [col sep=comma]{\string"storageSz-simCapAvgRateDeliveryMu.csv"};		
 		\addplot+table [col sep=comma]{\string"storageSz-simCapAvgRateDeliverySuNoCache.csv"};	 
 		\addplot+table [col sep=comma]{\string"storageSz-simCapAvgRateDeliverySu.csv"}; 							  			  		
	\end{axis}
\end{tikzpicture}
\vspace*{0.0cm}
\\
\scriptsize
\hspace{0.65cm}\ref{namedmethods4stogeoccov}
\caption{\small Evolution of average delivery rate in capacity-aided deployment. 
$\lambda_{\mathrm{cr}} = 1.0\times 10^{-5}$, $\lambda_{\mathrm{mc}} = 1.5\times 10^{-5}$, $\lambda_{\mathrm{sc}'} = 1.5\times 10^{-5}$, $\lambda_{\mathrm{ut-m}} = 3.0\times 10^{-5}$	unit/m$^2$;
$\hat{c} = 3$ units;
$P_{\mathrm{mc}} = 16$, $P_{\mathrm{sc}} = 3$ Watt;
$\tau_{\mathrm{mc}} = \tau_{\mathrm{sc}} = 4$ bits/s/Hz;
$\alpha = 4$;
$R_{\mathrm{c}} = 80$ meters;
$\mu = 30$ bits/s/Hz 
$\gamma = 0.6$;
$f_0 = 500$, $F_{\mathrm{sc}} = 4$ GByte;
$\eta = 1.45$.}
\label{fig:avgDeliveryRate:cap}
\vspace*{-0.55cm}
\end{figure}
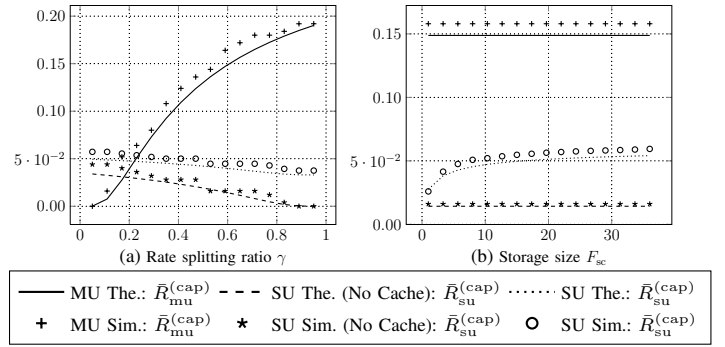